\documentclass[english,b4paper,12pt]{amsart}

\usepackage[margin=2cm]{geometry}

\usepackage{cite}

\usepackage{times}
\usepackage{euler}

   \usepackage[pdftex]{graphicx}
   \DeclareGraphicsExtensions{.pdf,.jpeg,.png}
\usepackage{amssymb,amsthm,stmaryrd}
\usepackage{amsmath}
\theoremstyle{plain}
\newtheorem{theorem}{Theorem}
\newtheorem*{theorem*}{Theorem}
\newtheorem{proposition}[theorem]{Proposition}
\newtheorem{corollary}[theorem]{Corollary}
\newtheorem*{corollary*}{Corollary}
\newtheorem{lemma}[theorem]{Lemma}

\theoremstyle{remark}
\newtheorem{remark}[theorem]{Remark}
\newtheorem*{remark*}{Remark}
\newtheorem{example}[theorem]{Example}

\theoremstyle{definition}
\newtheorem{definition}[theorem]{Definition}

\usepackage[noend]{algorithmic}
\usepackage{algorithm}

\usepackage{enumerate}
\usepackage{color}

\usepackage[matrix,arrow,curve]{xy}

\newcommand{\field}[1][]{\mathbb{F}_{#1}}
\newcommand{\bfield}{\field[]}

\DeclareMathOperator{\distance}{d}

\newcommand{\lclm}[1]{\left[#1\right]_\ell}
\newcommand{\lcrm}[1]{\left[#1\right]_r}
\newcommand{\gcrd}[1]{\left(#1\right)_r}

\begin{document}
%
\title{A Sugiyama-like decoding algorithm for convolutional codes}
%
%
%


\author{Jos\'{e} G\'{o}mez-Torrecillas,
        F.~J. Lobillo,
        Gabriel Navarro
}
\date{}
\thanks{Jos\'{e} G\'{o}mez-Torrecillas and F. J. Lobillo are with CITIC and Department of Algebra of University of Granada, Gabriel Navarro is with CITIC and Department of Computer Science and Artificial Intelligence of University of Granada.}

%



\maketitle

\begin{abstract}
We propose a decoding algorithm for a class of convolutional codes called skew BCH convolutional codes. These are convolutional codes of designed Hamming distance endowed with a cyclic structure yielding a left ideal of a non-commutative ring (a quotient of a skew polynomial ring). In this setting, right and left division algorithms exist, so our algorithm follows the guidelines of the Sugiyama's procedure for finding the error locator and error evaluator polynomials for BCH block codes.
\end{abstract}


\section{Introduction}
The main reason why cyclic block codes are useful is that it is possible to exploit the ring structure of their word-ambient space to get a better control of the parameters of the code, and to design efficient decoding algorithms. A classical example is the procedure developed by Sugiyama, Kasahara, Hirasawa and Namekawa \cite{Sugiyama} for nearest neighbor decoding of BCH codes.
Commonly known as Sugiyama Algorithm, it is a variation of the decoding scheme proposed by  Peterson \cite{Peterson} and, Gorenstein and Zierler \cite{Zierler}, which computes the error positions of a received polynomial by a clever use of the extended Euclidean algorithm. 

When dealing with convolutional codes, the Viterbi algorithm is, by far, the most often used for decoding convolutional codes over binary symmetric or additive white Gaussian noise channels. It makes use of the trellis structure of these codes in order to find the shortest path and return a maximum-likelihood estimation by means of hard and soft decission schemes. Tom\'{a}s, Rosenthal and Smarandache \cite{TRS11} use large finite windows in the infinite sliding generating matrix associated to convolutional codes to design a decoding algorithm over the erasure channel. It is known that endowing convolutional codes with a cyclic structure requires of a non-commutative multiplication \cite{Piret}. However, the different proposals of cyclic convolutional codes in the literature seem to have failed to take advantage of their algebraic structure for finding efficient and practical decoding algorithms, aiming to provide an alternative to the Viterbi algorithm. This is probably due to the fact that the non-commutative polynomial rings used in this classical approach \cite{Piret,Roos,Heide} are more complicated than expected. In particular, no Euclidean division algorithm is available here.   In \cite{GLNSCCC} we proposed a simpler approach that follows the idea of Piret \cite{Piret} and Roos \cite{Roos} of using a non-commutative multiplication, but implements it with a different algebraic construction. Thus, a skew cyclic convolutional code (SCCC) becomes a left ideal, whose generator is expressed in a adequate way,  of a suitable factor ring of a skew polynomial ring with coefficients in a rational function field.

This paper is the natural continuation of \cite{GLNSCCC}. We mainly attempt to show more solid evidences of the great potential of this notion, and provide a decoding algorithm for a class of SCCCs. Hopefully, this could lay the foundations of a practical alternative of the Viterbi algorithm. By analogy with BCH codes, we call these codes skew BCH convolutional codes. The similarities with the block case allow us to design a Sugiyama-like algorithm for decoding them.

The paper is organized as follows. Section \ref{BCH} is devoted to fix the algebraic setting and notation we will use.  We define skew convolutional BCH codes making use of the construction developed in \cite[Section IV]{GLNSCCC}. Our presentation of the generators of these codes as least common left multiples of sets of linear skew polynomials allows us to   prove that skew BCH convolutional codes are MDS with respect to the Hamming distance and provide bounds for their free distance. In Section \ref{sugiyama} we define the error locator and error evaluator polynomials, and we prove that they satisfy a non-commutative key equation. Since, in this setting, left and right division algorithms are available, we solve the key equation making use of the Right Extended Euclidean Algorithm (REEA). Unlike the classical block case, for less errors than the error-correcting capacity of the code, our method can fail solving the key equation. In Section \ref{failures} we shall prove that the theoretical probability of a key equation failure is zero, or, in practice, that it tends to zero as the maximum degree of the coefficients goes towards to infinity. In despite of this, we shall also give a subsidiary procedure that can be executed whenever the REEA fails to solve the key equation, which outputs the error locator and error evaluator polynomials. Finally, in order to make the paper more self-contained and, for the sake of readers non-familiar with skew polynomial rings, we have added an Appendix containing basic information about them. We have also moved there some technical results for making the paper more readable.

All along the paper, the theory is illustrated by examples. These have been implemented and computed with the aid of mathematical software SageMath \cite{sage}.

\section{Skew BCH convolutional codes}\label{BCH}

Let us first fix the notation we shall use throughout the paper. Let $\mathbb{F}=\mathbb{F}_q$ be the field with $q$ elements, where $q$ is a power of a prime number, and  $\mathbb{F}(t)$ denote the field of rational functions over $\mathbb{F}$, i.e. the field of fractions of the polynomial ring $\mathbb{F}[t]$. Consider an $\mathbb{F}$--algebra automorphism $\sigma$ of $\mathbb{F}(t)$ of (necesarily) finite order \(n\). For brevity, we denote by $R$ the (non-commutative) ring of skew polynomials $\mathbb{F}(t)[x;\sigma]$, that is, the $\mathbb{F}(t)$-vector space of standard (commutative) polynomials whose product is skewed by the rule $x\gamma=\sigma(\gamma)x$ for any $\gamma \in \mathbb{F}(t)$,  see the Appendix for details on this ring. The polynomial $x^n-1$ is central in $R$, and,  therefore, the left ideal of $R$ it generates is two-sided, so we may consider the quotient  ring \(\mathcal{R} = \bfield(t)[x;\sigma] / \langle x^n-1 \rangle\), which is isomorphic, as an $\mathbb{F}(t)^\sigma$--algebra, to the matrix ring $\mathcal{M}_n(\mathbb{F}(t)^\sigma)$ over the field of invariants  $\mathbb{F}(t)^\sigma$, see \cite[Theorem 1]{GLNSCCC}. A \emph{skew cyclic convolutional code} (SCCC) $\mathcal{C}$ is defined as a convolutional code whose preimage $\mathfrak{v}^{-1}(\mathcal{C})$ via the coordinate map $\mathfrak{v}:\mathcal{R}\to \mathbb{F}(t)^n$ is a left ideal of $\mathcal{R}$. Here, we are taking coordinates with respect to the basis $\{1, x, \dots, x^{n-1}\}$ (modulo $x^n-1$) of $\mathcal{R}$ considered as an $\mathbb{F}(t)$--vector space. For simplicity, unless otherwise stated, we shall  identify $\mathcal{C}$ with  $\mathfrak{v}^{-1}(\mathcal{C})$ and say that an SCCC is a left ideal of $\mathcal{R}$. A method for constructing SCCCs of fixed dimension is given in \cite[Section IV]{GLNSCCC}. Concretely, by the Normal Basis Theorem, we may choose an element $\alpha\in \mathbb{F}(t)$ such that $\{\alpha,\sigma(\alpha),\ldots , \sigma^{n-1}(\alpha)\}$ is a basis of $\mathbb{F}(t)$ as an $\mathbb{F}(t)^\sigma$-vector space. We may set then $\beta=\alpha^{-1}\sigma(\alpha)$, which satisfies the property
\begin{equation}\label{mainprop}\lclm{x-\beta,x-\sigma(\beta),\ldots ,x-\sigma^{n-1}(\beta)}=x^n-1,\end{equation}
where $\lclm{-}$ denotes the least common left multiple in $R$. Under these conditions, let \( \{i_1 < \dots < i_k\}\subset \{0,1,\ldots ,n-1\}\) be a set of indices, the code $\mathcal{C}$ generated as a left ideal by \(f = [x-\sigma^{i_1}(\beta), x-\sigma^{i_2}(\beta), \ldots ,x-\sigma^{i_k}(\beta)]_\ell\) is an SCCC of length $n$ and dimension $n-k$. Of course, when we say that a left ideal of $\mathcal{R}$ is generated by some $f \in R$, this generator has to be understood as the equivalence class of $f$ modulo $\langle x^n -1 \rangle$. We will keep this abuse of language all along the paper.

 The purpose of this section is to give a systematic method for constructing SCCCs of a designed Hamming distance. Due to the analogy with BCH block codes, we shall call them \emph{skew BCH convolutional codes}.  
The following technical result, which is a particular case of \cite[Corollary 4.13]{Lam/Leroy:1988}, is of importance in the sequel. We include an elementary proof.

\begin{lemma}\cite[Corollary 4.13]{Lam/Leroy:1988}\label{circulantlemma}
Let \(L\) be a field, $\sigma$ an automorphism of $L$ of finite order \(n\), and \(K = L^\sigma\) the invariant subfield under $\sigma$. Let \(\{\alpha_0, \dots, \alpha_{n-1}\}\) be a \(K\)--basis of \(L\). Then, for all \(t \leq n\) and every subset \(\{k_0 < k_1 < \dots < k_{t-1}\} \subseteq \{0, 1, \dots, n-1\}\),
\[
\begin{vmatrix}
\alpha_{k_0} & \alpha_{k_1} & \dots & \alpha_{k_{t-1}} \\
\sigma(\alpha_{k_0}) & \sigma(\alpha_{k_1}) & \dots & \sigma(\alpha_{k_{t-1}}) \\
\vdots & \vdots & \ddots & \vdots \\
\sigma^{t-1}(\alpha_{k_0}) & \sigma^{t-1}(\alpha_{k_1}) & \dots & \sigma^{t-1}(\alpha_{k_{t-1}})
\end{vmatrix} \neq 0.
\]
\end{lemma}

\begin{proof}
We prove the statement by induction on \(t\). The case \(t = 1\) holds trivially. So assume that the lemma is satisfied for some \(t \geq 1\). We need to check that, for any \((t+1)\times(t+1)\)--matrix
\[
\Delta = \begin{pmatrix}
\alpha_{k_0} & \alpha_{k_1} & \dots & \alpha_{k_{t}} \\
\sigma(\alpha_{k_0}) & \sigma(\alpha_{k_1}) & \dots & \sigma(\alpha_{k_{t}}) \\
\vdots & \vdots & \ddots & \vdots \\
\sigma^{t}(\alpha_{k_0}) & \sigma^{t}(\alpha_{k_1}) & \dots & \sigma^{t}(\alpha_{k_{t}})
\end{pmatrix},
\]
the determinant \(|\Delta|\) is non zero. Suppose, contrary to this, that \(|\Delta | = 0\). By the induction hypothesis, the first \(t\) columns of \(\Delta\) are linearly independent, so there exist \(a_0, \dots, a_{t-1} \in L \) such that the last column
\[
(\alpha_{k_{t}}, \sigma(\alpha_{k_{t}}), \dots, \sigma^{t}(\alpha_{k_{t}})) = \sum_{j=0}^{t-1} a_j (\alpha_{k_{j}}, \sigma(\alpha_{k_{j}}), \dots, \sigma^{t}(\alpha_{k_{j}})).
\]
That is, \(a_0, \dots, a_{t-1}\) satisfy the linear system
\begin{equation}\label{S}
\left\{\begin{aligned}
\alpha_{k_t} &= a_0 \alpha_{k_0} + a_1 \alpha_{k_1} + \dots + a_{t-1} \alpha_{k_{t-1}} \\
\sigma(\alpha_{k_t}) &= a_0 \sigma(\alpha_{k_0}) + a_1 \sigma(\alpha_{k_1}) + \dots + a_{t-1} \sigma(\alpha_{k_{t-1}}) \\
& \vdots \\
\sigma^{t}(\alpha_{k_t}) &= a_0 \sigma^{t}(\alpha_{k_0}) + a_1 \sigma^{t}(\alpha_{k_1}) + \dots + a_{t-1} \sigma^{t}(\alpha_{k_{t-1}}).
\end{aligned}\right.
\end{equation}
For any \(j = 0, \dots, t-1\), we subtract in \eqref{S} the equation \(j+1\) transformed by \(\sigma^{-1}\)  from  the equation $j$. This yields the homogeneous linear system
\begin{equation}\label{Sprime}
\left\{\begin{aligned}
0 &= (a_0 - \sigma^{-1}(a_0)) \alpha_{k_0} + (a_1 - \sigma^{-1}(a_1)) \alpha_{k_1} + \dots + (a_{t-1} - \sigma^{-1}(a_{t-1})) \alpha_{k_{t-1}} \\
0 &= (a_0 - \sigma^{-1}(a_0)) \sigma(\alpha_{k_0}) + (a_1 - \sigma^{-1}(a_1)) \sigma(\alpha_{k_1}) + \dots + (a_{t-1} - \sigma^{-1}(a_{t-1})) \sigma(\alpha_{k_{t-1}}) \\
& \vdots \\
0 &= (a_0 - \sigma^{-1}(a_0)) \sigma^{t-1}(\alpha_{k_0}) + (a_1 - \sigma^{-1}(a_1)) \sigma^{t-1}(\alpha_{k_1}) + \dots + (a_{t-1} - \sigma^{-1}(a_{t-1})) \sigma^{t-1}(\alpha_{k_{t-1}}).
\end{aligned}\right.
\end{equation}
The coefficient matrix of \eqref{Sprime} is non singular, by the induction hypothesis, so, for all \(j = 0, \dots, t-1\), \(a_j - \sigma^{-1}(a_j) = 0\), and hence \(a_0, \dots, a_{t-1} \in K\). Consequently, the first equation of \eqref{S} provides a linear dependence over \(K\) of the $K$-basis \(\{\alpha_0, \dots, \alpha_{n-1}\}\), a contradiction. Thus \(|\Delta | \neq 0\) and the lemma is proved. 
\end{proof}

We recall the reader, see for instance \cite[pp. 310]{Lam/Leroy:1988}, that, for any \(\gamma \in \bfield(t)\), the \emph{$j$-th norm} of $\gamma$ is defined to be $$N_j(\gamma) = \gamma \sigma(\gamma) \dots \sigma^{j-1}(\gamma).$$ By Lemma \ref{eval} in the Appendix, the remainder of the left division of a polynomial \(g = \sum_{i=0}^{r} g_i x^i\) by \(x-\gamma\) is \(\sum_{i=0}^r g_i N_i(\gamma)\). Whenever $x-\gamma$ right divides $g$, we shall say that $\gamma$ is a \emph{right root} of $g$. The notion of $j$-norm also admits a version for negative numbers given by
$$N_{-j}(\gamma)=a\sigma^{-1}(\gamma)\cdots \sigma^{-j+1}(\gamma).$$
Then, the remainder of the right division of a polynomial \(g = \sum_{i=0}^{r} g_i x^i\) by \(x-\gamma\) can be written as \(\sum_{i=0}^r \sigma^{-i}(g_i) N_{-i}(\gamma)\). Whenever $x-\gamma$ left divides $g$ we say that $\gamma$ is a \emph{left root} of $g$. For any integers $i$ and $j$, $N_i(\sigma^j(\gamma))=\sigma^j(N_i(\gamma))$, see Lemma \ref{eval} in the Appendix.

\begin{lemma}\label{degreelclm}
Let $\alpha\in \mathbb{F}(t)$ such that $\{\alpha,\sigma(\alpha),\ldots , \sigma^{n-1}(\alpha)\}$ is a basis of $\mathbb{F}(t)$ as an $\mathbb{F}(t)^\sigma$-vector space. Set $\beta=\alpha^{-1}\sigma(\alpha)$.  For any subset $T=\{t_1<t_2<\cdots <t_m\}\subseteq \{0,1,\ldots ,n-1\}$,  the polynomials $$g^{\ell}=\lclm{x-\sigma^{t_1}(\beta),x-\sigma^{t_2}(\beta),\ldots ,x-\sigma^{t_m}(\beta)} \text{ and } g^{r}=\lcrm{x-\sigma^{t_1}(\beta^{-1}),x-\sigma^{t_2}(\beta^{-1}),\ldots ,x-\sigma^{t_m}(\beta^{-1})}$$ have degree $m$. Consequently, if $x-\sigma^{s}(\beta)\mid_{r} g^{\ell}$ or $x-\sigma^{s}(\beta^{-1})\mid_{\ell} g^{r}$, then $s\in T$.
\end{lemma}
\begin{proof}
Let us suppose that $\deg g^{\ell}<m$, so that  \(g^{\ell} = \sum_{i=0}^{m-1}g_i x^i\). Since \(g\) is a left multiple of \(x - \sigma^{t_j}(\beta)\) for any \(1 \leq j \leq m\), it follows from Lemma \ref{eval}\,$i)$ in the Appendix that
\begin{equation}\label{linsystNN}
\sum_{i=0}^{m-1} g_i N_i(\sigma^{t_j}(\beta)) = 0 \text{ for any } 1 \leq j \leq m.
\end{equation}
This is a homogeneous linear system whose coefficient matrix is the transpose of
\[
M = \begin{pmatrix}
N_0(\sigma^{t_1}(\beta)) & N_0(\sigma^{t_2}(\beta)) & \dots & N_0(\sigma^{t_m}(\beta)) \\
N_1(\sigma^{t_1}(\beta)) & N_1(\sigma^{t_2}(\beta)) & \dots & N_1(\sigma^{t_m}(\beta)) \\
N_2(\sigma^{t_1}(\beta)) & N_2(\sigma^{t_2}(\beta)) & \dots & N_2(\sigma^{t_m}(\beta)) \\
\vdots & \vdots & \ddots & \vdots \\
N_{m-1}(\sigma^{t_1}(\beta)) & N_{m-1}(\sigma^{t_2}(\beta)) & \dots & N_{m-1}(\sigma^{t_m}(\beta)) \\
\end{pmatrix}.
\]Note that $N_i(\sigma^{t_j}(\beta))=\sigma^{t_j}(N_i(\beta))=\sigma^{t_j}(\alpha^{-1})\sigma^{t_j+i}(\alpha))$ for any $1 \leq j \leq m$ and $0 \leq i \leq m-1$. Thus, $|M|=0$ if and only if the determinant of the matrix 
$$M'=\begin{pmatrix}
\sigma^{t_1}(\alpha) & \sigma^{t_2}(\alpha) & \dots & \sigma^{t_m}(\alpha)) \\
\sigma^{t_1+1}(\alpha) & \sigma^{t_2+1}(\alpha) & \dots & \sigma^{t_m+1}(\alpha)) \\
\sigma^{t_1+2}(\alpha) & \sigma^{t_2+2}(\alpha) & \dots & \sigma^{t_m+2}(\alpha)) \\
\vdots & \vdots & \ddots & \vdots \\
\sigma^{t_1+m-1}(\alpha) & \sigma^{t_2m-1}(\alpha) & \dots & \sigma^{t_m+m-1}(\alpha)) \\
\end{pmatrix}=
\begin{pmatrix}
\sigma^{t_1}(\alpha) & \sigma^{t_2}(\alpha) & \dots & \sigma^{t_m}(\alpha)) \\
\sigma(\sigma^{t_1}(\alpha)) & \sigma(\sigma^{t_2}(\alpha))& \dots & \sigma(\sigma^{t_m}(\alpha)) \\
\sigma^2(\sigma^{t_1}(\alpha)) & \sigma^2(\sigma^{t_2}(\alpha))& \dots & \sigma^2(\sigma^{t_m}(\alpha)) \\
\vdots & \vdots & \ddots & \vdots \\
\sigma^{m-1}(\sigma^{t_1}(\alpha)) & \sigma^{m-1}(\sigma^{t_2}(\alpha))& \dots & \sigma^{m-1}(\sigma^{t_m}(\alpha)) \\\end{pmatrix}
$$
is zero. However, by Lemma \ref{circulantlemma}, $|M'|\not =0$, so the single solution of the linear system \eqref{linsystNN} is \(g_0 = \dots = g_{m-1} = 0\), a contradiction. Therefore \(\deg g^{\ell} = m\).  For the other polynomial we proceed similarly: if $\deg g^r<m$ and $g^r=\sum_{i=0}^{m-1}g_ix^i$, we obtain the linear system
\begin{equation}\label{linsystNN2}
\sum_{i=0}^{m-1} \sigma^{-i}(g_i) N_{-i}(\sigma^{t_j}(\beta^{-1})) = 0 \text{ for any } 1 \leq j \leq m.
\end{equation}
Observe now that $N_{-i}(\sigma^{t_j}(\beta^{-1}))=\sigma^{t_j}(\alpha^{-1})\sigma^{t_j-i+1}(\alpha)$ for $0\leq i \leq m-1$ and $1\leq j \leq m$. Then, by Lemma \ref{circulantlemma}, the system has a single solution $\sigma^{-i}(g_i)=0$ for $0\leq i\leq m-1$, so $g_0=g_1=\cdots =g_{m-1}=0$. Again, this yields a contradiction, so $\deg g^r=m$.
\end{proof}

\begin{definition}\label{BCHCC}
Let $\alpha,\beta\in\mathbb{F}(t)$ verifying the conditions of Lemma \ref{degreelclm}. A skew BCH convolutional code of designed distance $\delta\leq n$ is an SCCC generated by $\lclm{x-\sigma^{r}(\beta),x-\sigma^{r+1}(\beta),\ldots , x-\sigma^{r+\delta-2}(\beta)}$, for some $r\geq 0$.
\end{definition}

\begin{theorem}\label{hdist}
Let $\mathcal{C}$ be a skew BCH convolutional code of designed distance $\delta$. The  Hamming distance of $\mathcal{C}$ is $\delta$.
\end{theorem}
\begin{proof}
Let us denote by $g=\lclm{x-\sigma^{r}(\beta),x-\sigma^{r+1}(\beta),\ldots , x-\sigma^{r+\delta-2}(\beta)}$, a generator of $\mathcal{C}$ as a left ideal of $\mathcal{R}$. A parity check matrix is 
\[
H = \begin{pmatrix}
N_0(\sigma^{r}(\beta)) & N_0(\sigma^{r+1}(\beta)) & \dots & N_0(\sigma^{r+\delta-2}(\beta)) \\
N_1(\sigma^{r}(\beta)) & N_1(\sigma^{r+1}(\beta)) & \dots & N_1(\sigma^{r+\delta-2}(\beta)) \\
N_2(\sigma^{r}(\beta)) & N_2(\sigma^{r+1}(\beta)) & \dots & N_2(\sigma^{r+\delta-2}(\beta)) \\
\vdots & \vdots & \ddots & \vdots \\
N_{n-1}(\sigma^{r}(\beta)) & N_{n-1}(\sigma^{r+1}(\beta)) & \dots & N_{n-1}(\sigma^{r+\delta-2}(\beta)) \\
\end{pmatrix},
\]
since its columns give the right evaluations on the roots. We have to prove that any $\delta-1$-minor of $H$ is non zero. We proceed similarly to the proof of Lemma \ref{degreelclm}. Note that  $N_i(\sigma{k}(\beta))=\sigma^{k}(N_i(\beta))=\sigma^{k}(\alpha^{-1})\sigma^{i+k}(\alpha)$ for any integers $i$ and $k$. Therefore, given a  submatrix of order $\delta-1$,
\[M = \begin{pmatrix}
N_{k_1}(\sigma^{r}(\beta)) & N_{k_1}(\sigma^{r+1}(\beta)) & \dots & N_{k_1}(\sigma^{r+\delta-2}(\beta)) \\
N_{k_2}(\sigma^{r}(\beta)) & N_{k_2}(\sigma^{r+1}(\beta)) & \dots & N_{k_2}(\sigma^{r+\delta-2}(\beta)) \\
N_{k_3}(\sigma^{r}(\beta)) & N_{k_3}(\sigma^{r+1}(\beta)) & \dots & N_{k_3}(\sigma^{r+\delta-2}(\beta)) \\
\vdots & \vdots & \ddots & \vdots \\
N_{k_{\delta-1}}(\sigma^{r}(\beta)) & N_{k_{\delta-1}}(\sigma^{r+1}(\beta)) & \dots & N_{k_{\delta-1}}(\sigma^{r+\delta-2}(\beta)) \\
\end{pmatrix} \, \text{with $\{k_1<k_2<\cdots <k_{\delta-1}\}\subset\{0,1,\ldots ,n-1\}$},
\]
$|M|=0$ if and only if $|M'|=0$, where $M'$ is the matrix
\[\begin{pmatrix}
\sigma^{k_1+r}(\alpha) & \sigma^{k_1+r+1}(\alpha) & \dots & \sigma^{k_1+r+\delta-2}(\alpha) \\
\sigma^{k_2+r}(\alpha) & \sigma^{k_2+r+1}(\alpha) & \dots & \sigma^{k_2+r+\delta-2}(\alpha) \\
\sigma^{k_3+r}(\alpha) & \sigma^{k_3+r+1}(\alpha) & \dots & \sigma^{k_3+r+\delta-2}(\alpha) \\
\vdots & \vdots & \ddots & \vdots \\
\sigma^{k_{\delta-1}+r}(\alpha) & \sigma^{k_{\delta-1}+r+1}(\alpha) & \dots & \sigma^{k_{\delta-1}+r+\delta-2}(\alpha) \\
\end{pmatrix}=\begin{pmatrix}
\sigma^{k_1+r}(\alpha) & \sigma(\sigma^{k_1+r}(\alpha)) & \dots & \sigma^{\delta-2}(\sigma^{k_1+r}(\alpha)) \\
\sigma^{k_2+r}(\alpha) & \sigma(\sigma^{k_2+r}(\alpha)) & \dots & \sigma^{\delta-2}(\sigma^{k_2+r}(\alpha)) \\
\sigma^{k_3+r}(\alpha) & \sigma(\sigma^{k_3+r}(\alpha)) & \dots & \sigma^{\delta-2}(\sigma^{k_3+r}(\alpha)) \\
\vdots & \vdots & \ddots & \vdots \\
\sigma^{k_{\delta-1}+r}(\alpha) & \sigma(\sigma^{k_{\delta-1}+r}(\alpha)) & \dots & \sigma^{\delta-2}(\sigma^{k_{\delta-1}+r}(\alpha)) \\
\end{pmatrix}.
\]
Since $\{\alpha,\sigma{(\alpha)},\ldots ,\sigma^{n-1}(\alpha)\}$ is a basis of the extension $\mathbb{F}(t)^\sigma\subset\mathbb{F}(t)$, by Lemma \ref{circulantlemma}, $|M'|\not =0$.
\end{proof}

\begin{corollary}
Let $\mathcal{C}$ be a skew BCH convolutional code of length $n$ and dimension $k$. Then  $$n-k+1\leq \mathrm{d}_{\mathrm{free}}(\mathcal{C}) \leq (n-k)(\lfloor m/k+1) \rfloor+m+1,$$
where $\mathrm{d}_{\mathrm{free}}(\mathcal{C})$ is the free distance of $\mathcal{C}$ and $m$ is its total memory.
\end{corollary}
\begin{proof}
By Theorem \ref{hdist}, a skew BCH convolutional code of designed Hamming distance $\delta$ has dimension $k=n-\delta+1$. Hence, $\delta=n-k+1$.  Now, the free distance of $\mathcal{C}$ is always greater than its Hamming distance, and lower than the generalized singleton bound $(n-k)(\lfloor m/k+1) \rfloor+m+1$. 
\end{proof}

In the following example we show that these bounds cannot be improved.

\begin{example}
Let $\mathbb{F}=\mathbb{F}_8$ be the field with eight elements, $\mathbb{F}(t)$ the field of rational functions over $\mathbb{F}$ and $\sigma:\mathbb{F}(t)\to \mathbb{F}(t)$ the automorphism defined by $\sigma(t)=1/t$, whose order is clearly two. The quotient algebra is then $\mathcal{R}=\mathbb{F}(t)[x;\sigma]/\langle x^2-1 \rangle$. Following the construction described in Definition \ref{BCHCC}, we consider the element $\alpha=t\in \mathbb{F}(t)$. Since
$$\begin{vmatrix}
\alpha & \sigma(\alpha)  \\ 
\sigma(\alpha) & \alpha 
\end{vmatrix}=\begin{vmatrix}
t & 1/t  \\ 
1/t & t 
\end{vmatrix}=t^2+1/t^2\not = 0,$$ $\{\alpha,\sigma(\alpha)\}$ is basis of the field extension $\mathbb{F}(t)^\sigma\subset\mathbb{F}(t)$. Set $\beta=\alpha^{-1}\sigma(\alpha)=1/t^2$. Hence $x-\beta=x+1/t^2$ gives a generator of an SCCC  $\mathcal{C}$ of length $n=2$ and dimension $k=1$. Actually, it is a skew BCH convolutional code of designed distance $\delta=2$. A basic (and minimal) generator matrix of $\mathcal{C}$ is given by $M=(1 \text{  }  t^2)$, so the degree (or total memory) $m$ of the encoder is 2. Now, following the procedure described in \cite[Theorems 3.4 and 3.6]{Zigangirov}, the three first terms of the sequence of column distances is 1,1,2; whilst the first terms of the sequence of row distances is 2,2,2. Therefore the free distance of $\mathcal{C}$, $\distance_{\mathrm{free}}(\mathcal{C} )=2=n-k+1$. Observe that the generalized singleton bound \cite{SmarandacheGluesingRosenthal} with these parameters is 6. 

 Let us now consider the automorphism  of order two defined by $\sigma(t)=2t$ on $\mathbb{F}=\mathbb{F}_3$. Set $\alpha=t+1$. In this case, $\beta=(2t + 1)/(t + 1)$ and $x-\sigma(\beta)$ is a generator of a skew BCH convolutional code with $\delta=2$, $n=2$ and $k=1$. A minimal generator matrix is $M=(t + 1 \text{  }  t + 2)$, so $m=1$. For this  SCCC, the $0$th column distance and the 3rd row distance are 4, hence $\distance_{\mathrm{free}}(\mathcal{C} )=4$. Then $n-k+1=2<4=\distance_{\mathrm{free}}(\mathcal{C} )=(n-k)(\lfloor m/k \rfloor+1)+m+1$, so it reaches the generalized singleton bound.
\end{example}

\section{A Sugiyama-like decoding algorithm}\label{sugiyama}

Throughout this section $\mathcal{C}$ will denote a skew BCH convolutional code of designed distance $\delta$ generated, as a left ideal of $\mathcal{R}$, by $g = \lclm{x-\sigma^{r}(\beta),x-\sigma^{r+1}(\beta),\ldots, x-\sigma^{r+\delta-2}(\beta)}$ for some $r\geq 0$, where $\beta$ is chosen as in Definition \ref{BCHCC}.  The Hamming distance of $\mathcal{C}$ is exactly $\delta$ (Theorem \ref{hdist}).  Set $\tau=\lfloor \frac{\delta-1}{2} \rfloor$ which is the maximum number of errors than the code can correct.
For simplicity, we shall suppose that $r=0$. This is not a restriction, because we may always write $\beta'=\sigma^{r}(\beta)$. Then $\beta'=(\alpha')^{-1}\sigma(\alpha')$, where $\alpha'=\sigma^r(\alpha)$, and $\alpha'$ also provides a normal basis. Therefore, $g=\lclm{x-\beta',x-\sigma(\beta'),\ldots, x-\sigma^{\delta-2}(\beta')}$. 

Let \(c \in \mathcal{C}\) be a codeword that is transmitted through a noisy channel and the polynomial \(y = c + e\) is received, where \(e = e_1 x^{k_1} + \dots + e_\nu x^{k_\nu}\) with \(\nu \leq \tau\).
We define the \emph{error locator} polynomial as  
\[
\lambda = \lcrm{1-\sigma^{k_1}(\beta)x, 1-\sigma^{k_2}(\beta) x, \ldots , 1-\sigma^{k_{\nu}}(\beta)x}.
\]
We first show that $\lambda$ determines the positions with a non-zero error. 

\begin{lemma}\label{bothlcrm}
For any subset $\{t_1,t_2,\ldots ,t_m\}\subseteq \{0,1,\ldots ,n-1\}$,
$$\lcrm{1-\sigma^{t_1}(\beta)x, 1-\sigma^{t_2}(\beta)x,\ldots ,1-\sigma^{t_m}(\beta)x}=\lcrm{x-\sigma^{t_1-1}(\beta^{-1}),x-\sigma^{t_2-1}(\beta^{-1}),\ldots ,x-\sigma^{t_m-1}(\beta^{-1})}.$$
\end{lemma}
 \begin{proof}
 For any $a\in \mathbb{F}(t)$, $1-ax=(x-\sigma^{-1}(a^{-1}))(-\sigma^{-1}(a))$ and $x-\sigma^{-1}(a^{-1})=(1-ax)(-\sigma^{-1}(a^{-1}))$. Therefore, the polynomials of the statement left divide one to each other.  \end{proof}

\begin{proposition}\label{ev}
 $1-\sigma^d(\beta)x$ left divides $\lambda$ if and only if $x-\sigma^{d-1}(\beta^{-1})$ left divides $\lambda$ if and only if $d\in \{k_1,\ldots, k_\nu\}$
 \end{proposition}
\begin{proof}
By Lemma \ref{bothlcrm}, $1-\sigma^d(\beta)x$ left divides $\lambda$ if and only if $x-\sigma^{d-1}(\beta^{-1})$ left divides $\lambda$.  Now, by Lemma \ref{degreelclm}, $x-\sigma^{d-1}(\beta^{-1})$ left divides $\lambda$ if and only if $d\in \{k_1,\ldots, k_\nu\}$.
\end{proof}

Therefore, once $\lambda$ is known, the error positions can be located by following the rule: $d\in \{0,1,\ldots ,n-1\}$ is an error position if and only if $\sigma^{d-1}(\beta^{-1})$ is a left root of $\lambda$. Observe that $\lambda$ may be replaced by any polynomial in $R$ associated on the right to $\lambda$, that, is, any polynomial differing from $\lambda$ by multiplication on the right by a nonzero element in $\mathbb{F}(t)$.

For any \(1 \leq j \leq \nu\), $\lambda= (1-\sigma^{k_j}(\beta) x) p_j$ for some polynomial $p_j\in R$ with \(\deg p_j = \nu-1\). We define the \emph{error evaluator} polynomial as 
$\omega = \sum_{j=1}^\nu e_j \sigma^{k_j}(\alpha) p_j$. Once we know the error locator polynomial and the error evaluator polynomial, we may compute the values $e_1,e_2,\ldots , e_\nu$ by solving a linear system and determine completely the error $e$. Observe also that \(\deg \omega < \nu\).

Finally, for each \(0 \leq i \leq n-1\), the \(i\)-th syndrome $S_i$ of the received polynomial $y  = \sum_{j=0}^{n-1}y_jx^j$ is defined to be the remainder of the left quotient of \(y\) by \(x -\sigma^i(\beta)\).  Observe that $S_i$ is the right evaluation of $y$ at $\sigma^i(\beta)$. Whenever $0\leq i \leq 2\tau-1$, the right evaluations on \(c\) are zero, and it follows that
\begin{equation}\label{syndromei}
S_i = \sum_{j=0}^{n-1} y_j N_j(\sigma^i(\beta)) = \sum_{j=1}^\nu e_j N_{k_j}(\sigma^i(\beta)) = \sum_{j=1}^\nu e_j \sigma^i(N_{k_j}(\beta))=\sum_{j=1}^\nu e_j \sigma^i(\alpha^{-1})\sigma^{k_j+i}(\alpha)=\sigma^i(\alpha^{-1})\sum_{j=1}^\nu e_j \sigma^{k_j+i}(\alpha).
\end{equation}
Therefore $\sigma^i(\alpha)S_i =\sum_{j=1}^\nu e_j \sigma^{k_j+i}(\alpha)$ and 
we call \(S= \sum_{i=0}^{2\tau-1} \sigma^i(\alpha)S_i x^i\)  the \emph{syndrome polynomial} of $y$. 

\begin{theorem}
The error locator and the error evaluator satisfy the \emph{non-commutative key equation}
\[
\omega = S \lambda + x^{2\tau}  u,
\]
where $u \in R$ is of degree less than $\nu$.
\end{theorem}

\begin{proof}
First, observe that $R = \mathbb{F}(t)[x;\sigma]$ may be seen as a subring of the skew power series ring $\mathbb{F}(t)[[x;\sigma]]$ (see, e.g., \cite[Chapter 1, Section 4]{McConnell/Robson:1988}).  Given \(1-ax \in \mathbb{F}(t)[x;\sigma] \) with $a \in \mathbb{F}(t)$, a straigtforward computation in  $\mathbb{F}(t)[[x;\sigma]]$ shows that \((1-ax)^{-1} = \sum_{i \geq 0} N_i(a) x^i\). Thus,  \(p_j =(1- \sigma^{k_j}(\beta)x)^{-1}\lambda= \sum_{i \geq 0} N_i( \sigma^{k_j}(\beta)) x^i \lambda\) for any $1\leq j \leq \nu$. Then,
\[
\begin{split}
\omega &= \sum_{j=1}^\nu e_j \sigma^{k_j}(\alpha) \sum_{i \geq 0} N_i(\sigma^{k_j}(\beta)) x^i \lambda \\
&= \sum_{i \geq 0} \Big(\sum_{j=1}^\nu e_j \sigma^{k_j}(\alpha) N_i(\sigma^{k_j}(\beta))\Big) x^i \lambda \\
&= \sum_{i \geq 0} \Big(\sum_{j=1}^\nu e_j \sigma^{k_j}(\alpha) \sigma^{k_j}(N_i(\beta))\Big) x^i \lambda \text{, by Lemma \ref{eval}} \\
&= \sum_{i \geq 0} \Big(\sum_{j=1}^\nu e_j \sigma^{k_j}(\alpha) \sigma^{k_j}(\alpha^{-1}\sigma^i(\alpha))\Big) x^i \lambda \\
&= \sum_{i \geq 0} \Big(\sum_{j=1}^\nu e_j \sigma^{k_j+i}(\alpha)\Big) x^i \lambda \\
&= \sum_{i=0}^{2\tau-1} \Big(\sum_{j=1}^\nu e_j \sigma^{k_j+i}(\alpha)\Big) x^i \lambda + \sum_{i \geq 2\tau} \Big(\sum_{j=1}^\nu e_j \sigma^{k_j+i}(\alpha)\Big) x^i \lambda \\
&= \sum_{i=0}^{2\tau-1} \sigma^i(\alpha)S_i x^i \lambda + x^{2\tau}\sum_{h \geq 0} \Big(\sum_{j=1}^\nu \sigma^{-2\tau}(e_j) \sigma^{k_j+h}(\alpha)\Big) x^h \lambda \text{, by (\ref{syndromei})} \\
&=S \lambda + x^{2\tau}\sum_{j=1}^\nu \sigma^{-2\tau}(e_j)  \sum_{h \geq 0} \sigma^{k_j+h}(\alpha) x^h \lambda \\
&=S \lambda + x^{2\tau}\sum_{j=1}^\nu \sigma^{-2\tau}(e_j)\sigma^{k_j}(\alpha)  \sum_{h \geq 0} N_h(\sigma^{k_j}(\alpha)) x^h \lambda \\
&=S \lambda + x^{2\tau}\sum_{j=1}^\nu \sigma^{-2\tau}(e_j)\sigma^{k_j}(\alpha)  (1-\sigma^{k_j}(\alpha)x)^{-1} \lambda \\
&=S \lambda + x^{2\tau}\sum_{j=1}^\nu \sigma^{-2\tau}(e_j)\sigma^{k_j}(\alpha) p_j \\
&=S \lambda + x^{2\tau}u, \\
\end{split}
\]
where $u=\sum_{j=1}^\nu \sigma^{-2\tau}(e_j)\sigma^{k_j}(\alpha) p_j$.
\end{proof}

We now proceed to solve the key equation. Concretely, we shall use a Sugiyama-like procedure for this task, i.e, we shall make use of the Right Euclidean Extended Algorithm (REEA), see the Appendix for details. We recall that, for any $f,g\in R$, each step $i$ of the REEA provides coefficients $\{u_i,v_i,r_i\}$ (the Bezout coefficients) such that $f u_i+gv_i=r_i$, where $(f,g)_{\ell}=r_h$ and $\deg r_{i+1}<\deg r_i$ for any $0\leq i\leq h-1$.

\begin{theorem}\label{thkeyeq}
The non-commutative key equation 
\begin{equation} \label{keyeq} x^{2\tau}u+S\lambda=\omega\end{equation}
 is a right multiple of the equation 
\begin{equation} \label{euclideq}x^{2\tau} u_I+Sv_I=r_I, \end{equation} where $u_I,v_I$ and $r_I$ are the Bezout coefficients returned by the REEA with input $x^{2\tau}$ and $f$, and $I$ is the index determined by the conditions $\deg r_{I-1}\geq \tau$ and $\deg r_I<\tau$. In particular, $\lambda=v_Ig$ and $\omega=r_Ig$ for some $g\in R$.
\end{theorem}
\begin{proof}
We recall that $\deg S<2\tau$, $\deg \lambda \leq \tau$ and $\deg \omega<\nu \leq \tau$, and, consequently, $\deg u<\tau$. On the other hand, by Lemma \ref{REEA}\,$vi)$ in the Appendix, $\deg v_I+\deg r_{I-1}=2\tau$, so that $\deg v_I\leq \tau$.

Let us consider the least common right multiple $\lcrm{\lambda,v_I}=\lambda a=v_Ib$, where $a,b\in R$ with $\deg a\leq \deg v_I\leq \tau$ and $\deg b\leq \deg \lambda \leq \tau$. Then $(a,b)_r=1$. Hence, we multiply (\ref{keyeq}) by the right by $a$, and  (\ref{euclideq}) on the right by $b$, to obtain \begin{equation} \label{keyeq2} x^{2\tau}ua+S\lambda a=\omega a\end{equation}  and \begin{equation} \label{euclideq2}x^{2\tau} u_Ib+Sv_Ib=r_Ib.\end{equation}
Hence, from (\ref{keyeq2}) and (\ref{euclideq2}), $x^{2\tau}(ua-u_Ib)=\omega a-r_Ib$. Comparing degrees, it follows that $ua=u_Ib$ and $\omega a=r_Ib$. Actually, $(a,b)_r=1$ yields $\lcrm{u,u_I}=ua=u_Ib$ and  $\lcrm{\omega,r_I}=\omega a=r_Ib$. In particular, $\deg a\leq \deg r_I<\tau$.

Let $\lclm{a,b}=a'a=b'b$. Since $\lcrm{\lambda,v_I}$ is a left multiple of $a$ and $b$, there exists $m\in R$ such that $\lcrm{\lambda,v_I}=m\lclm{a,b}$. Then $\lambda a=v_I b=ma'a=mb'b$. Thus, $\lambda=ma'$ and $v_I=mb'$ and, by minimality, $(\lambda,v_I)_{\ell}=m$. Similar arguments prove that there exists $m',m''\in R$ such that $u_I=m'b'$ and $u=m'a'$, and that $r_I=m''b'$ and $\omega=m''a'$. Nevertheless, by Lemma \ref{REEA}\,$v)$ in Appendix, $(u_I,v_I)_r=1$, so $b'=1$. In this way, $b=a'a$ and  we get $\lambda=v_Ia'$, $\omega=r_Ia'$ and $u=u_Ia'$. This completes the proof.
\end{proof}

Observe that if $(\lambda, \omega)_r = 1$, then Theorem \ref{thkeyeq} gives an algorithmic procedure to compute both the error locator and the error evaluator polynomials. However, unlike the classical (commutative) block case, these non-commutative polynomials could have a non-trivial common right divisor, as Example \ref{decfa} below shows. Nevertheless, we will show latter that in most cases $(\lambda,\omega)_r = 1$ (Theorem \ref{gcrdtodet}). Therefore, Algorithm \ref{Sug} will rarely fail to decode.
\begin{example}\label{decfa}
Let $\mathbb{F}=\mathbb{F}_8$ be the field of eight elements generated over $\mathbb{F}_2$ by a primitive element $a$ with $a^3 + a + 1 = 0$. For brevity, except for 0 and 1, we shall write the elements of $\mathbb{F}$ as powers of $a$. Let $\sigma:\mathbb{F}(t)\to \mathbb{F}(t)$ be the automorphism defined by $\sigma(t)=(t+a)/t$. The order of $\sigma$ is 7, so, in this case, the sentence-ambient algebra is $\mathcal{R}=\mathbb{F}(t)[x;\sigma]/\langle x^7-1 \rangle$. The element $\alpha=t$ yields a normal basis $\{\alpha,\sigma(\alpha),\ldots , \sigma^6(\alpha)\}$. Set $\beta=\alpha^{-1}\sigma(\alpha)=(t+a)/t^2$. Let $\mathcal{C}$ be the skew BCH convolutional code generated by $g=\lclm{\{x-\sigma^i(\beta)\}_{i=0,1,2,3}}$, whose Hamming distance is, in virtue of Theorem \ref{hdist}, $\delta=5$, and it corrects up to $\tau=2$ errors. Suppose we receive a polynomial $y$ such that the error to be removed is $e=1+ x$, i.e. there are errors at positions $k_1=0$ and $k_2=1$, whose values are both $e_1=e_2=1$. In such a case, the error locator polynomial is $$\lcrm{1-\beta x,1-\sigma(\beta) x}=x^{2} + \left(\frac{t^{3}}{a^3 t^{3} + a^3t^{2} + a^{2} t + a^{2}}\right) x + \frac{a t^{3} + a^4 t}{t^{4} + a t^{3} +t^{2} + a^2 t +a^4},$$
and the error evaluator polynomial is as follows:
$$\omega=t\,p_0(x)+\sigma(t)\, p_1(x)=\left(\frac{ t^{2} +  t + a}{ t +
 1}\right) x + \frac{a^4 t^{4} + a^4 t^{3} + a^4 t^{2} + t + a}{a^3 t^{4} + a^4 t^{3} + a^3 t^{2}
+a^5 t + 1}.$$
Now, we may compute the greatest common right divisor $$(\lambda,\omega)_r=x + \frac{a t^{2} + a^4}{t^{3} + a^3
t^{2} + a t + a^4}
 \not =1.$$
Thus, in this case, only a left divisor of the error locator polynomial $\lambda$ is computed by the REEA. In other words, we cannot deduce all positions of the error.
\end{example}

\begin{algorithm}[h]
\caption{\texttt{Decoding algorithm for skew BCH convolutional codes}}\label{Sug}
\begin{algorithmic}[1]
\REQUIRE A received polynomial $y=\sum_{i=0}^{n-1}y_ix^i$ obtained from the transmission of a codeword $c$ in a skew BCH convolutional code $\mathcal{C}$ generated by $g=\lclm{\{x-\sigma^i(\beta)\}_{i=0, \ldots , \delta-2}}$ of error-correcting capacity $\tau=\lfloor \frac{\delta-1}{2}\rfloor$.
\ENSURE A codeword $c'$, or \emph{key equation failure}.
\FOR{\(0 \leq i \leq 2\tau-1\)}
	\STATE $S_i\gets\sum_{j=0}^{n-1} y_jN_j(\sigma^i(\beta))$
\ENDFOR
\STATE $S\gets \sum_{i=0}^{2\tau-1} \sigma^i(\alpha)S_ix^i$
\IF{$S=0$}
	\RETURN $y$
\ENDIF
\STATE $\{u_i,v_i,r_i\}_{i=0,\ldots , l}\gets \text{REEA}(x^{2\tau},S)$
\STATE $I\gets$ first iteration in REEA with $\deg r_i<\tau$
\STATE $pos\gets \emptyset$
\FOR{\(0 \leq i \leq n-1\)}
	\IF{$\sigma^{i-1}(\beta^{-1})$ is a left root of $v_I$}
		\STATE $pos=pos  \cup  \{i\}$
	\ENDIF
\ENDFOR
\IF{$\deg v_I > \mathrm{Cardinal}(pos)$}\label{kef}
	\RETURN \emph{key equation failure}
\ENDIF
\FOR{\(j\in pos\)}
	\STATE $p_j\gets \operatorname{right-quotient}(v_I,1-\sigma^j(\beta)x)$
\ENDFOR
\STATE Solve the linear system $r_I=\sum_{j\in pos}e_j\sigma^{j}(\alpha) p_j$
\STATE $e\gets \sum_{j\in pos} e_jx^j$
\RETURN $y-e$
\end{algorithmic}
\end{algorithm}

\begin{remark}
Algorithm \ref{Sug} fails to decode when the condition of Line \ref{kef} in Algorithm \ref{Sug} is fulfilled, as a consequence of Lemma \ref{degreelclm}. As we shall prove in Theorem \ref{criterion}, this condition is equivalent to $\deg v_I<\deg \lambda$. Therefore, no further key equation failure can be expected. As discussed above, this key equation failure will happen rarely. Nevertheless, we will discuss how to solve it in Section \ref{failures}. In this way, Algorithm \ref{Sug} will be completed to a full decoding algorithm.
\end{remark}

 Next example illustrates a successful application of Algorithm \ref{Sug}.

\begin{example}
Under the conditions of Example \ref{decfa}, let us suppose that we receive the polynomial 
$$y=x^{4} + \left(\frac{a^2 t + 1}{a^5
t^{4} + a^3}\right) x^{3} + \frac{t^{6} + a^{2} t^{5} + t^{4} +
t^{3} + a^6 t}{a^5 t^{6} +
a^4 t^{5} + a^5 t^{4} +
a^3 t^{2} + a^{2} t + a^3}.$$
This is just the generator of the code in which we have removed the coefficients of degree 1 and 2. Therefore, there are errors at two positions and our algorithm should correct them. We first compute the syndrome polynomial
\begin{displaymath}\begin{split}
S= & \left(\frac{a^{2} t^{7} + t^{6} + a^3 t^{5} + t^{3} +
t^{2} + a^3t + a}{a^{2} t^{6} + a^3 t^{5}
+ a^5 t^{4} + t^{2} + a t + a^3}\right) x^{3} +
\left(\frac{a^4 t^{7} + t^{6} + a^{2} t^{5} +
a^5 t^{4} + a^{2} t^{2} + t + a}{a^6 t^{7} + a t^{6} + a t^{5} + a^6t^{4} +
a^4 t^{3} + a^6 t^{2} +
a^6 t + a^4}\right) x^{2} \\
+  & \left(\frac{a^{2}
t^{5} + a^6 t^{4} + a t^{3} + a^6
t^{2} + a^{2}}{a^5 t^{5} + a t^{4} + a^3 t + a^6}\right) x +
 \frac{a^4 t^{6} +
a^4 t^{5} + a^{2} t^{4} +a^4
t^{3} + t^{2} +a^5 t + a}{a^{2} t^{5} +
a^5 t^{4} + t + a^3}.
\end{split}
\end{displaymath}
We now apply REEA until we get a reminder of degree less than $\tau=2$, and
$$v_I=x^{2} + \left(\frac{a t^{3} + a^{2} t^{2} +a^3 t + a^4}
{t^{2} + 1}\right) x + \frac{a^6 t^{4} + t^{3} +
t^{2} + a t}{a^{2} t^{3} + a^{2} t^{2} + a t + a}.$$
The left roots of $v_I$ in the set $\{\sigma^i(\beta^{-1})\}_{i=0,\ldots , 6}$ are $\sigma^0(\beta^{-1})$ and $\sigma^1(\beta^{-1})$, so, as we expect, there are errors at positions 1 and 2. On the other hand,
$$r_I=\left(\frac{t^{9} + a^5 t^{7} + a^3 t^{6} + a^5 t^{5} + a^6 t^{4} + a^6 t^{3} + t^{2} + t + a^6}{a^4 t^{7} + a^5 t^{6} +
t^{5} + a^{2} t^{3} + a^3 t^{2} + a^5 t}\right) x +
 \frac{a^5 t^{10} + a t^{9}
+ a^3t^{8} + a^6 t^{7} + a^3 t^{5} + t^{3} + t^{2} + a^4 t}{a^6 t^{8} + t^{7} +a^3 t^{6} + a^6t^{5} + a^{2} t^{4} + a^5 t^{3} + a
t^{2} + a^4 t + a^6}.
$$
We now solve the linear system $r_I=e_1\sigma(t)\,p_1+e_2\sigma^2(t)\,p_2$ and compute the values of the errors. Concretely, 
$$e_1=       	
\frac{a^{2} t^{7} + a t^{6} + a t^{5} + a^4 t^{3} +
a^3 t^{2} + a^3 t}{a^6
t^{6} + t^{5} + a^{2} t^{4} +a^4 t^{2} +a^5 t + 1} \text{ and }
e_2=\frac{a^{2} t^{6} + a t^{5} + a t^{4} + a^6 t^{2} +
a^5 t + a^5}{t^{5} + t^{4} +
a^5t + a^5}.$$
Therefore, $e=e_1x+e_2x^2$ and the received polynomial is correctly decoded to $y+e=g$.
\end{example}

\section{Key equation failures}\label{failures}

In this section we focus on the problem of a key equation failure. Obviously, the main questions to answer is how  often such a failure can occur and if, despite of this, we still may recover the error locator polynomial. Firstly, we have to point out that a single error is always corrected. All along the section we follow the notation of Section \ref{sugiyama}.

\begin{lemma}\label{l1}
$\deg v_I \geq 1$. As a consequence, if $\deg \lambda=1$ then $v_I$ and $\lambda$ are right associated.
\end{lemma}
\begin{proof}
The proof follows from the degrees of the polynomials. Indeed, by Theorem \ref{thkeyeq},  $\omega=r_I g$ for some $g\in R$, so $\deg g<\nu$. Since $\lambda=v_Ig$, by Lemma \ref{degreelclm}, $\nu=\deg \lambda=\deg v_I+\deg g$. Thus $\deg v_I\geq 1$.
\end{proof}

We now deal with the problem of computing an error locator polynomial once a key equation failure occurs. Following Algorithm \ref{Sug}, by means of the execution of the REEA, polynomials $r_I, u_I, v_I \in R$ satisfying the equality $x^{2t}u_I+Sv_I=r_I$ with $\lambda=v_Ig$ and $\omega=r_Ig$ for some $g\in R$. Moreover, $\deg v_I \geq 1$ by Lemma \ref{l1}. If $\deg g = 0$, then $v_I$ serves as a locator polynomial, and Algorithm \ref{Sug} will correctly decode the received polynomial. Our strategy when $\deg g > 0$ will consist in finding an increasing chain of right divisors of $\lambda$ whose first piece is $v_I$.  First, we prove a criterion to decide whether or not the error locator polynomial is reached. 

\begin{theorem}\label{criterion}
Let $q,p,s\in R$ such that $x^{2\tau}q+Sp=s$, $qg=u$, $pg=\lambda$ and $sg=\omega$ for some $g\in R$. Let $T=\{t_1, t_2,\cdots ,t_m\}\subset \{0,1,\ldots,n-1\}$ be the set of indices verifying that $\sigma^{j-1}(\beta^{-1})$ is a  left root of $p$ if and only if $j\in T$.
Then $m=\deg p$ if and only if $g$ is a constant.
\end{theorem}
\begin{proof}
We reorder the set of error positions in such a way that $T=\{k_1,\ldots ,k_m\}$ with $m\leq \nu$. If $\deg g=0$, $m=\nu$ and $\deg p=\nu$, by Lemma \ref{degreelclm}. Conversely, if $m=\deg p$, then $$p=\lcrm{1-\sigma^{k_1}(\beta)x,\ldots , 1-\sigma^{k_m}(\beta)x}
$$ by Lemma \ref{degreelclm} and Lemma \ref{bothlcrm}. Write $p=(1-\sigma^{k_j}(\beta)x)p'_j$ for any $j=1,\ldots ,m$. By Lemma \ref{locatorandp_j}\,$iii)$ in the Appendix, each polynomial of degree less than $m$ can be written as an $\mathbb{F}(t)$-linear combination of the polynomials $p'_1,\ldots ,p'_{m-1}, p'_m$. In particular, since $\deg s=\deg \omega-\deg g=\deg \omega+\deg p-\deg \lambda\leq \nu-1+m-\nu=m-1$, we get $s=\sum_{i=1}^m a_ip'_i$ for some $a_{1},\ldots, a_{m}\in \mathbb{F}(t)$. On the other hand, $\lambda=pg$. Thus, for any $j=1,\ldots ,m$, $(1-\sigma^{k_j}(\beta)x)p_j=(1-\sigma^{k_j}(\beta)x)p'_jg$, so $p_j=p'_jg$. 
Now, $sg=\omega$, so 
\begin{equation}\label{aux}
 \sum_{j=1}^m a_jp_j= \left (\sum_{j=1}^m a_jp'_j \right )g=sg=\omega=\sum_{j=1}^m e_j\sigma^{k_j}(\alpha)p_j+\sum_{j=m+1}^\nu e_{j}\sigma^{k_j}(\alpha)p_{j}.
\end{equation}
By Lemma 26\,$iv)$ in the Appendix, $\{p_1,\ldots ,p_\nu\}$  gives a basis of $R/R\lambda$  as an $\mathbb{F}(t)$--vector space.  Therefore, since $e_i\sigma^{k_i}(\alpha) \neq 0$ for every $i \leq \nu$, equation \eqref{aux}  implies that
$m=\nu$ and, thus,  $\deg g=0$.
\end{proof}

Whenever, in Algorithm \ref{Sug}, a key equation failure occurs, we may execute Algorithm \ref{find_root} in order to find a new error position.

\begin{algorithm}[h]
\caption{\texttt{Find-a-position}}\label{find_root}
\begin{algorithmic}[1]
\REQUIRE A non-constant polynomial $p$ with $\lambda=pg$ for some $g\in R$, $pos=\{i\geq 0 \text{ with } (1-\sigma^{i}(\beta)x)\mid_{\ell} p\}$, with $\deg p>\mathrm{Cardinal}(pos)$.
\ENSURE $d\notin pos$ such that $(1-\sigma^d(\beta)x)$ left divides $\lambda$.
\STATE $f\gets p$, $e \gets \deg f$
\FOR{$0\leq i \leq n-1$}
	\IF{$i\notin pos$}
		\STATE $f\gets \lcrm{f,1-\sigma^i(\beta)x}$
		\IF{$\deg f = e$}
			\RETURN $i$
		\ELSE
			\STATE $e\gets e+1$
		\ENDIF		
	\ENDIF
\ENDFOR
\end{algorithmic}
\end{algorithm}

\begin{proposition}
Algorithm \ref{find_root} correctly finds a new error position.
\end{proposition}
\begin{proof}
Let $T=\{t_1<t_2<\cdots <t_r\}=\{0,1,\ldots ,n-1\}\setminus pos$. For any $0\leq i\leq r$, we denote $\lambda_i=\lcrm{\lambda_{i-1},1-\sigma^{t_i}(\beta)x}$ with $\lambda_{0}=\lambda$, and $f_i=\lcrm{f_{i-1},1-\sigma^{t_i}(\beta)x}$ with $f_{0}=p$. It is clear that $f_i\mid_{\ell} \lambda_i$ for any $i=0,\ldots, r$. 
We prove first that the algorithm must return a position. Suppose that the sequence $\{\deg f_i\}_{0\leq i\leq r}$ always grows. Hence $\deg f_{r}=r+\deg p>n-\mathrm{Cardinal}(pos)+\deg p>n$. This is not possible, since $f_{r}\mid_{\ell} \lambda_{r}=x^n-1$. So there exists a minimal $d\geq 0$ such that $\deg f_{d-1}=\deg f_d$. Now, $1-\sigma^{t_d}(\beta)x \mid_{\ell} f_{d-1} \mid_{\ell} \lambda_{d-1}=\lcrm{\lambda,1-\sigma^{t_1}(\beta)x,\ldots ,1-\sigma^{t_{d-1}}(\beta)x}$. Since, $t_d\not =t_1,\ldots ,t_{d-1}$, $1-\sigma^{t_d}(\beta)x \mid_{\ell} \lambda$. Thus, by Lemma \ref{degreelclm} and Corollary \ref{ev}, $t_d$ is an error position.
\end{proof}

Therefore, by means of a recursive application of Algorithm \ref{find_root}, we may find all error positions, and both the error locator and the error evaluator polynomials, see Algorithm \ref{key_solver}.

\begin{algorithm}[h]
\caption{\texttt{Key equation failure solver}}\label{key_solver}
\begin{algorithmic}[1]
\REQUIRE  Polynomials $v_I, r_I$ with $\lambda=v_Ig$, $\omega=r_Ig$ for some $g\in R$, the set $pos=\{i\geq 0 \text{ with } (1-\sigma^{i}(\beta)x)\mid_{\ell} v_I\}$
\ENSURE The error locator polynomial $\lambda$ and the error evaluator polynomial $\omega$.
\STATE $f\gets v_I$, $h\gets r_I$
\WHILE{$\mathrm{Cardinal}(pos)<\deg f$}
	\STATE $d\gets \operatorname{Find-a-position}(f,pos)$
	\STATE $f \gets \lcrm{f,1-\sigma^d(\beta)x}$
	\FOR{$0\leq i \leq n-1$}
		\IF{$i\notin pos$ and $1-\sigma^i(\beta)x\mid_{\ell} f$}
			\STATE $pos\gets pos\cup \{i\}$
		\ENDIF
	\ENDFOR	
\ENDWHILE
\STATE $g\gets \operatorname{rquot}(f,v_I)$
\RETURN $f$, $hg$
\end{algorithmic}
\end{algorithm}

\begin{example}
Consider the code and received word of Example \ref{decfa}. As we have seen there, there is a key equation failure since $(\lambda,\omega)_r\not =1$. Actually, by applying REEA, the polynomials $v_I=x + t/(t+1)$ and $r_I=(t^2+t+a)/(t+1)$, and the set of known error positions $pos$ is the empty set. We follow Algorithm \ref{key_solver} and compute $$\lcrm{v_I,1-\sigma^0(\beta)x}=x^{2} + \left(\frac{t^{3}}{a^3
t^{3} + a^3 t^{2} + a^2 t + a^2}\right) x + \frac{a t^{3} + a^4
t}{t^{4} + a t^{3} + t^{2} + a^{2} t + a^4}.$$
The degree has grown so we continue and compute $\lcrm{v_I,1-\sigma^0(\beta)x, 1-\sigma^1(\beta)x}$. Fortunately, in this case, the reader may check that the degree remains  being two. 
By Algorithm \ref{find_root}, 1 is an unknown error position. Therefore, $$v^{(1)}_I\gets\lcrm{v_I, 1-\sigma^1(\beta) x}=x^{2} + \left(\frac{t^{3}}{a^3
t^{3} + a^3 t^{2} + a^2 t + a^2}\right) x + \frac{a t^{3} + a^4
t}{t^{4} + a t^{3} + t^{2} + a^{2} t + a^4}.$$ is a left divisor of the error locator polynomial. Next, we must update the set of known error positions. One can see that  $v^{(1)}_I=\lclm{1-\sigma^0(\beta)x, 1-\sigma^1(\beta)x}$, so $pos=\{0,1\}$ and $\mathrm{Cardinal}(pos)=\deg v^{(1)}_I$. Thus, by Theorem \ref{criterion}, $\lambda=v^{(1)}_I$. Now, the right quotient of $\lambda$ over $v_I$, $$g=x + \frac{a^4 t^{2} + 1}{a^3 t^{3} +
a^6 t^{2} + a^4 t + 1}$$ and  $$\omega=r_Ig=\left(\frac{ t^{2} +  t + a }{ t + 1}\right) x +
\frac{\left(a^{2} + 1\right) t^{4} + \left(a^{2} + 1\right) t^{3} +
\left(a^{2} + 1\right) t^{2} + a^{2} t + a + 1}{\left(a^{2} + a +
1\right) t^{4} + \left(a^{2} + 1\right) t^{3} + \left(a^{2} + a +
1\right) t^{2} + t + a^{2}},$$
so the error locator and the error evaluator polynomials are determined.
\end{example}

 We close this section analyzing how often a key equation failure occurs. Indeed, for a given set of error positions, we will show that the values of the errors must satisfy a non-trivial relation. Recall that such a failure is only possible if $(\lambda, \omega)_r \neq  1$.
 
\begin{proposition}\label{gcrdtodet}
\(\gcrd{\omega,\lambda} = 1\) if and only if
\begin{equation}\label{detno0}
\begin{vmatrix}
e_1 & e_2  & \dots & e_\nu \\
\sigma(e_1) & \sigma(e_2) & \dots & \sigma(e_\nu) \\
\vdots & \vdots & \ddots & \vdots \\
\sigma^{\nu-1}(e_1) & \sigma^{\nu-1}(e_2) & \dots & \sigma^{\nu-1}(e_\nu)\end{vmatrix}\neq 0.
\end{equation}
\end{proposition}

\begin{proof}
By Proposition \ref{cyclicgenerator} and Lemma \ref{coordinates} in the Appendix, \(\gcrd{\omega,\lambda} = 1\) if and only if the matrix 
\[
A = \begin{pmatrix}
e_1 \sigma^{k_1}(\alpha)& e_2 \sigma^{k_2}(\alpha)& \dots & e_\nu \sigma^{k_\nu}(\alpha)\\
\sigma(e_1) \sigma^{k_1}(\alpha)& \sigma(e_2) \sigma^{k_2}(\alpha)& \dots & \sigma(e_\nu) \sigma^{k_\nu}(\alpha)\\
\vdots & \vdots & \ddots & \vdots \\
\sigma^{\nu-1}(e_1) \sigma^{k_1}(\alpha)& \sigma^{\nu-1}(e_2) \sigma^{k_2}(\alpha)& \dots & \sigma^{\nu-1}(e_\nu)\sigma^{k_\nu}(\alpha)
\end{pmatrix} \]
has full rank. Since  $\sigma^{k_j}(\alpha) \neq 0$ for every $j = 1, \dots, \nu$, we get that the determinant of $A$ is non-zero if and only if  the determinant in \eqref{detno0} is not zero.
\end{proof}

\begin{theorem}\label{thprob}
If \(\gcrd{\omega,\lambda} \neq 1\) then the error values \(e_1, \dots, e_\nu\) are linearly dependent over \(\mathbb{F}(t)^{\sigma}\).
\end{theorem}

\begin{proof}
 Assume that  \(\{e_1, e_2, \dots, e_\nu \}\) are linearly independent over $\mathbb{F}(t)^{\sigma}$, and let \(u_{\nu+1}, \dots, u_{n} \in \mathbb{F}(t)\) be such that \[\{e_1, e_2, \dots, e_\nu, u_{\nu+1}, \dots, u_{n}\}\] is an \(\mathbb{F}(t)^\sigma\)--basis of \(\mathbb{F}(t)\). But then Lemma \ref{circulantlemma} implies the inequality \eqref{detno0}. By Proposition \ref{gcrdtodet}, $(\lambda,\omega)_r = 1$. 
\end{proof}

 \begin{remark}
Taking coordinates with respect to a fixed basis of $\mathbb{F}(t)$ as an ($n$--dimensional) vector spcace over $K =\mathbb{F}(t)^{\sigma}$, we deduce from Theorem \ref{thprob} that the set of errors $\{ e_1, \dots, e_\nu \}$ giving a key equation failure is contained in the determinantal algebraic sub-variety of $K^{\nu n}$ determined by the common zeroes of all $\nu \times \nu$ minors. The  dimension of this varety is known to be at most $n-\nu + 1$ (see, e.g., \cite[Exercise 10.10]{Eisenbud:1995}), which is strictly smaller than $\nu n$ if $\nu > 1$.   Consequently, the theoretical probability that a key equation failure occurs is zero. 
\end{remark}

 \begin{remark}
 Skew block codes were defined in \cite{Boucher} and \cite{Ulmer} as left ideals of a factor ring of a skew polynomial ring  of the form $\mathbb{F}[x;\sigma]$. The results and algorithms of our paper will work almost verbatim in this setting, that is,  if we consider an automorphism  $\sigma$ of the finite field $\mathbb{F}$ instead of an $\mathbb{F}$-automorphism of $\mathbb{F}(t)$. This would then represent an alternative construction and decoding scheme for some of the skew cyclic block codes described in \cite{Boucher} and \cite{Ulmer}. Obviously, in this case, the theoretical probability of a key equation failure is always positive. 
\end{remark}

\section{Conclusions}

In this paper we have designed a decoding algorithm for convolutional codes aiming to provided the basis of a future alternative to the celebrated Viterbi algorithm. The algorithm uses the algebra structure of a class of codes, named skew BCH convolutional codes by the authors, in order to follow a Sugiyama-like procedure for determining the position of the errors. For a number of errors less than the error-correcting capacity of the code, the probability of a key equation failure tends to zero as the maximum degree of the coefficients goes towards to infinity. An auxiliary algorithm is designed for resolving any key equation failure and, henceforth, compute a error locator polynomial even in this case.

\section*{Appendix}
 In this appendix, we prove some technical facts needed in the paper. We also collect some basic facts 
on the  the non-commutative polynomial ring $\mathbb{F}(t)[x;\sigma]$ for the convenience of the reader non-familiar with the theory of Ore extensions (or skew polynomial rings).  The general theory was systematized in \cite{Ore:1933}. A good introduction of its basics essentials is the first chapter of \cite{Jacobson:1996}. 
For our purposes, we only need to consider the particular case of a skew polynomial ring constructed from a field automorphism. So, let $D$ be a field, and $\sigma$ an automorphism of $D$. The elements of the skew polynomial ring $R = D[x;\sigma]$ are standard polynomials in the indeterminate $x$ with coefficients in $D$ written on the left. The sum of polynomials in $R$ is defined as in the commuative case. The product is based on the rules $x^nx^m = x^{n + m}$ for $n, m \in \mathbb{N}$, while  $xa = \sigma (a) x$, for every $a \in D$. 

The degree $\deg(f)$ of a left polynomial $f \in R$, as well as its leading coefficient $\mathrm{lc}(f) \in D$,  are defined in the usual way. Hence $\deg (fg) = \deg(f) + \deg(g)$.
 The ring $R$ is a non-commutative domain, and there exist  both left and right Euclidean division algorithms, that work much as in the commutative case, with some adjustements coming from the non-commutativity. For intance, the right Euclidean algorithm is described in Algorithm \ref{rightdivalg}.  The polynomials $r$ and $q$ obtained as the output of Algorithm \ref{rightdivalg} are called \emph{right remainder} and \emph{right quotient}, respectively, of the right division of $f$ by $g$. We will use the notation $r=\mathrm{rrem}(f,g)$ and $q = \mathrm{rquot}{(f,g)}$. Analogous conventions and notations are used for the left division algorithm. 

\begin{algorithm}
\caption{\texttt{Right Euclidean Division}}\label{rightdivalg}
\begin{algorithmic}
\REQUIRE $f, g \in D[x;\sigma]$ with $g \neq 0$.
\ENSURE $q, r \in D[x;\sigma]$ such that $f = gq + r$ and $\deg r < \deg g$. 
\STATE $q\gets 0$, $r\gets f$
\WHILE{$\deg g \leq \deg r$}
	\STATE $a \gets \sigma^{-\deg g}(\mathrm{lc}(g)^{-1}\mathrm{lc}(r))$
	\STATE $q\gets q + ax^{\deg r - \deg g}$
	\STATE $r\gets r - g a x^{\deg r - \deg g}$ 
\ENDWHILE
\RETURN $q$, $r$
\end{algorithmic}
\end{algorithm}

\begin{remark}
There is no universal agreement in the literature in the  use of the adjetives ``left'' and ``right'' concerning the Euclidean division. For instance, what is called left Euclidean division algorithm in \cite{Jacobson:1996} is considered as the right one in \cite{Lam/Leroy:1988}. We follow Jacobson's convention.  
\end{remark}

 These  division algorithms allow to prove, in the usual way, that every left and every right ideal of $R$ is principal. The principal left ideal generated by a given $f \in R$ is denoted by $Rf$, while $fR$ denotes the principal right ideal generated by $f$. 

Given  nonzero $f, g \in R$,
 $Rf \subseteq Rg$ if and only if $g$ is \emph{right divisor} of $f$ or $f$ is a \emph{left multiple} of $g$. Moreover, $Rf = Rg$ if and only if there is a nonzero $u \in D$ such that $f = ug$. We say then that $f$ and $g$ are left associated. A standard argument shows that $Rf + Rg = Rd$ if and only if $d$ is the greatest common right divisor of $f$ and $g$. We will use the notation $d = \gcrd{f,g}$. It is uniquely determined up to left associates. Similarly,
$Rf \cap Rg = Rm$ if and only if $m$ is the lowest common left multiple of $f$ and $g$, which is denoted by $m = \lclm{f,g}$, which is unique up to left associates.  $\lclm{-}$. The equality $\deg \lclm{f,g}=\deg f+\deg g - \deg \gcrd{f,g}$ holds in this non-commutative setting. Both $\gcrd{f,g}$ and $\lclm{f,g}$ can be computed by using the Left Extended Euclidean Algorithm. The right side version of these definitions and properties can be stated analogously. For our purposes, we describe explicitly a version of the Right Extended Euclidean Algorithm which provides the Bezout's coefficients in each step of the algorithm, see Algorithm \ref{REEA_alg}.

\begin{algorithm}
\caption{\texttt{Right Euclidean Extended Algorithm}}\label{REEA_alg}
\begin{algorithmic}
\REQUIRE $f, g \in D[X;\sigma]$ with $f \neq 0, g \neq 0$.
\ENSURE $\{u_i,v_i,r_i\}_{i=0,\ldots ,h, h+1}$ such that $r_i=fu_i+gv_i$ for any $i$, $r_h=(f,g)_{\ell}$,  and $u_{h+1}f = [f,g]_r$.
\STATE $r_0\gets f$, $r_1\gets g$.
\STATE $u_0\gets 1$, $u_1\gets 0$.
\STATE $v_0\gets 0$, $v_1\gets 1$.
\STATE $q\gets 0$, $rem \gets 0$.
\STATE $i\gets 1$.
\WHILE{$r_i\not = 0$}
	\STATE $q, rem\gets \operatorname{rquot-rem}(r_{i-1},r_i)$
	\STATE $r_{i+1}\gets rem$
	\STATE $u_{i+1}\gets u_{i-1}-u_iq$
	\STATE $v_{i+1}\gets v_{i-1}-v_iq$
	\STATE $i\gets i+1$
\ENDWHILE
\RETURN $\{u_i,v_i,r_i\}_{i=0,\ldots ,h, h+1}$
\end{algorithmic}
\end{algorithm}

The following result is a right-side version for Ore polynomials of \cite[Lemma 3.8]{VonzurGathen}.

\begin{lemma}\label{REEA}
Let $f,g\in D[x;\sigma]$ and $\{u_i,v_i,r_i\}_{i=0,\ldots ,h}$ be the coefficients obtained when applying the REEA to $f$ and $g$. 
Let us denote $R_0=\left (\begin{smallmatrix} u_0 & v_0 \\ u_1 & v_1 \end{smallmatrix}\right )$,  $Q_i=\left (\begin{smallmatrix} 0 & 1 \\ 1 & -q_i \end{smallmatrix}\right )$ and $R_i=R_0Q_1\cdots Q_i$ for any $i=0,\ldots , h$. Hence, for any $i=0,\ldots , h$, the following items hold:
\begin{enumerate}[i)]
\item $(f \, g) R_i=(r_{i-1} \, r_i)$.
\item $R_i=\left (\begin{smallmatrix} u_i & u_{i+1} \\ v_{i} & v_{i+1} \end{smallmatrix}\right )$.
\item $f u_i+g v_i=r_i$.
\item $R_i$ has a left and right inverse.
\item $(u_i,v_i)_r=1$.
\item $\deg f=\deg r_{i-1}+\deg v_i$
\end{enumerate}
\end{lemma}
\begin{proof}
$i)$, $ii)$ and $iii)$ may be proven similarly to \cite[Lemma 3.8 $i)$, $ii)$ and $iv)$]{VonzurGathen}.  For $iv)$, observe that $T_i=\left (\begin{smallmatrix} q_i & 1 \\ 1 & 0 \end{smallmatrix}\right )$ is a left and right inverse of $Q_i$. So $S_i=T_i\cdots T_1R_0$  is a left and right inverse of $R_i$. Finally, for $v)$, if $S_i=\left (\begin{smallmatrix} a & b \\ c & d \end{smallmatrix}\right )$,
$$\left (\begin{matrix} a & b \\ c & d \end{matrix}\right )\left (\begin{matrix} u_i & u_{i+1} \\ v_{i} & v_{i+1} \end{matrix}\right )=\left (\begin{matrix} 1 & 0 \\ 0 & 1 \end{matrix}\right ),$$
so there exist $a,b\in D[x;\sigma]$ verifying $au_i+bv_i=1$. Thus $(u_i,v_i)_r=1$.

$vi)$. For $i=1$, $r_0=f$ and $v_1=1$, so the equality holds. Note that $\deg r_i < r_{i-1}$ and $\deg v_{i-1} < v_{i}$ for any $i > 1$. Then, since $r_{i+1}=r_{i-1}-r_iq_i$ and $v_{i+1}=v_{i-1}-v_iq_i$ for any $i$, $\deg r_{i-1}=\deg r_{i}+\deg q_i$ and $\deg v_{i+1}=\deg v_i+\deg q_i$ for any $i$. Now, by the induction hypothesis, $\deg f=\deg r_{i-1}+\deg v_i=\deg r_i+\deg q_i+\deg v_{i+1}-\deg q_i=\deg r_i+\deg v_{i+1}$. 
\end{proof}

From now on, set $R = \mathbb{F}(t)[x;\sigma]$.  The \emph{lef evaluation} of a non--commuative polynomial $g \in R$ at $\alpha \in \mathbb{F}(t)$ is the remainder of the left division of $g$ by $x-\alpha$, and similarly for the \emph{right evaluation}. These evaluations allows to speak of left and right roots of non-commutative polynomials. Their properties in a general setting were studied in \cite{Lam/Leroy:1988}.

\begin{lemma}\label{eval}
Let $\gamma\in \mathbb{F}(t)$ and $g=\sum_{i=1}^{r}g_ix^i\in R$. Then:
\begin{enumerate}[$i)$]
\item The remainder of the left division of $g$ by $x-\gamma$ is $\sum_{i=0}^rg_iN_i(\gamma)$
\item The remainder of the right division of $g$ by $x-\gamma$ is $\sum_{i=0}^r\sigma^{-i}(g_i)N_{-i}(\gamma)$
\item $N_j(\sigma^k(\gamma))=\sigma^k(N_j(\gamma))$ for any $i,k$.
\end{enumerate}
\end{lemma}
\begin{proof}
 $i)$ and $ii)$ are deduced from \cite[Lemma 2.4]{Lam/Leroy:1988}. \\
$iii)$. $N_j(\sigma^k(\gamma))=\sigma^k(\gamma)\sigma^{k+1}(\gamma)\cdots \sigma^{k+j-1}(\gamma)=\sigma^k(\gamma\sigma(\gamma)\cdots \sigma^{j-1}(\gamma))=\sigma^k(N_j(\gamma))$
\end{proof}

\begin{lemma}\label{locatorandp_j}
Let $\{t_1<t_2<\cdots <t_m\}\subseteq \{0,1,\ldots ,n-1\}$ with $m >1$, and $q=\lcrm{1-\sigma^{t_1}(\beta)x,1-\sigma^{t_2}(\beta)x, \ldots ,1-\sigma^{t_m}(\beta)x}$. Let $q_1,\ldots, q_m\in R$ such that $q=(1-\sigma^{t_j}(\beta)x)q_j$ for any $1\leq j\leq m$.
Then:
\begin{enumerate}[$i)$]
\item $\lclm{q_1,q_2,\ldots , q_m} = q$ and $\gcrd{q_1,q_2,\ldots , q_m} = 1$. 
\item  \(R/Rq = \bigoplus_{j=1}^m Rq_j/Rq\).
\item For any $f\in R$ with $\deg f<m$ there exist $a_1,\ldots ,a_m\in \mathbb{F}(t)$ such that $f=\sum_{j=1}^m a_jq_j$.
\item The set \(\{q_1, \dots, q_m\}\) gives modulo $Rq$ a basis of \(R/Rq\) as an $\mathbb{F}(t)$-vector space.
\end{enumerate}
\end{lemma}
\begin{proof}
$i)$ By Lemmas \ref{degreelclm} and \ref{bothlcrm}, $\deg q=m$ and, thus,  $\deg q_j=m-1$ for any $j=1,\ldots ,m$.  Since $m >1$,  the degree of $[q_1, \dots,q_m]_\ell$ must be at least $m-1+1 = m$. But $q$ is obviously a left common multiple of $q_1, \dots, q_m$, whence $q =[q_1, \dots,q_m]_\ell$.  \\
$ii)$ Since   \(R q \subseteq R q_j\) for all \(1 \leq j \leq m\) and $\gcrd{q_1,q_2,\ldots , q_m} = 1$, we get \(R/Rq = \sum_{j=1}^m Rq_j/Rq\). Observe that \(Rq_j/Rq \cong R/R(1-\sigma^{t_j}(\beta)x)\) is one-dimensional over $\mathbb{F}(t)$. Since the dimension of $R/Rq$ as an $\mathbb{F}(t)$--vector space is $\deg q = m$, we get  the direct sum.\\
 $iii)$ and  $iv)$ follow from $ii)$.

\end{proof}

\begin{proposition}\label{cyclicgenerator}
The following statements are equivalent:
\begin{enumerate}[$i)$]
\item \(\gcrd{\omega,\lambda} = 1\).
\item \(\omega + R\lambda\) generates \(R/R\lambda\) as left \(R\)--module.
\item The set  $\{ x^i(\omega + R\lambda) ~|~ 0 \leq i \leq \nu-1 \}$ is linearly independent over $\mathbb{F}(t)$. 
\end{enumerate}
\end{proposition}

\begin{proof}
The equivalence between $i)$ and $ii)$ is a direct consequence of Bezout's Theorem. It is clear that $\omega + R\lambda$ generates the left $R$--module $R/R\lambda$  if and only if  $\{ x^i(\omega + R\lambda) ~|~ 0 \leq i \leq \nu-1 \}$ spans $R/R\lambda$ as an $\mathbb{F}(t)$--vector space. Since the dimension over $\mathbb{F}(t)$ of $R/R\lambda$ is $\nu$, the equivalence between $ii)$ and $iii)$ becomes clear. 
\end{proof}

\begin{lemma}\label{coordinates}
The $j$-coordinate of \(x^i \omega + R\lambda\) with respect to $\{p_1, \dots, p_\nu\}$ is $\sigma^i(e_j)\sigma^{k_j}(\alpha)$, for any $1\leq j \leq \nu$\end{lemma}
\begin{proof}
First note that $R(1-\sigma^{t_j}(\beta)x)=R(x-\sigma^{t_j}(\beta^{-1}))$ for $j=1,\ldots ,m$. By Lemma \ref{eval}, $\sigma^{t_j}(\beta^{-1})$ is a right root of $x^i-N_i(\sigma^{t_j}(\beta^{-1}))$. Then $x^i-N_i(\sigma^{t_j}(\beta^{-1}))\in R(1-\sigma^{t_j}(\beta)x)$. Multiplying on the right by $p_j$, $x^ip_j-N_i(\sigma^{t_j}(\beta^{-1}))p_j\in R\lambda$. Thus, in $R/R\lambda$,
\[x^i \omega = \sum_{j=1}^\nu x^i e_j\sigma^{k_j}(\alpha) p_j = \sum_{j=1}^\nu \sigma^i(e_j) \sigma^{k_j+i}(\alpha)x^i p_j=\sum_{j=1}^\nu \sigma^i(e_j) \sigma^{k_j+i}(\alpha)N_i(\sigma^{k_j}(\beta^{-1})) p_j.
\]
Now, 
$$\sigma^i(e_j) \sigma^{k_j+i}(\alpha)N_i(\sigma^{k_j}(\beta^{-1}))=\sigma^i(e_j) \sigma^{k_j+i}(\alpha)\sigma^{k_j}(N_i(\beta^{-1})))=\sigma^i(e_j) \sigma^{k_j+i}(\alpha)\sigma^{k_j}(\alpha\sigma^i(\alpha^{-1}))=\sigma^i(e_j)\sigma^{k_j}(\alpha),$$
and the result follows.
\end{proof}


\end{document}